\documentclass[11pt]{article}

\usepackage{amsmath}
\usepackage{amsthm}
\usepackage{amssymb}
\usepackage{enumerate}
\usepackage{amssymb}
\usepackage{tikz}
\usepackage{hyperref}
\usepackage{mathtools}

\usepackage[top=1.4in, bottom=1.4in, left=1.2in, right=1.2in]{geometry}

\theoremstyle{definition} \newtheorem{definition}{Definition}[section]
\theoremstyle{definition} \newtheorem{remark}[definition]{Remark}
\theoremstyle{definition} 
\theoremstyle{definition} 
\theoremstyle{definition} 
\theoremstyle{definition} 
\theoremstyle{plain} \newtheorem{theorem}[definition]{Theorem}
\theoremstyle{plain} 
\theoremstyle{plain} \newtheorem{lemma}[definition]{Lemma}
\theoremstyle{plain} 
\theoremstyle{plain} \newtheorem{corollary}[definition]{Corollary}
\theoremstyle{definition} \newtheorem{observation}[definition]{Observation}
\theoremstyle{plain} \newtheorem{conjecture}[definition]{Conjecture}
\theoremstyle{definition}

\newcommand{\fff}{\mathbb{F}}
\newcommand{\fffbar}{\overline{\mathbb{F}}}
\newcommand{\qqq}{\mathbb{Q}}
\newcommand{\rrr}{\mathbb{R}}
\newcommand{\ccc}{\mathbb{C}}

\newcommand{\eee}{\mathrm e}

\newcommand{\calC}{\mathcal{C}}

\newcommand{\calR}{\mathcal{R}}

\newcommand{\eps}{\varepsilon}    

\newcommand{\wto}{\widetilde O}

\DeclareMathOperator{\rank}{rank}

\newcommand{\expv}{\mathbb{E}}
\newcommand{\prob}{\mathbb{P}}

\DeclareMathOperator{\Var}{Var}

\title{Improved upper bounds for the rigidity of Kronecker products\footnote{This paper is to appear in the proceedings of the 46th International Symposium on
Mathematical Foundations of Computer Science (MFCS'21)}}
\author{Bohdan Kivva \vspace{0.3cm}\\ 
University of Chicago \\
bkivva@uchicago.edu}

\begin{document}
\maketitle
\begin{abstract}
The \emph{rigidity} of a matrix $A$ for target rank $r$ is the minimum
number of entries of $A$ that need to be changed in order to obtain a
matrix of rank at most $r$.  At MFCS'77, Valiant introduced matrix rigidity as a tool to prove circuit lower bounds for linear functions and since then this notion received much attention and found applications in other areas of complexity theory. The problem of constructing an explicit family of matrices that are sufficiently rigid for Valiant's reduction (Valiant-rigid) still remains open. 
 Moreover, since 2017 most of the long-studied candidates
have been shown not to be  Valiant-rigid.

Some of those former candidates for rigidity are Kronecker
products of small matrices.  In a recent paper (STOC'21),
Alman gave a general non-rigidity result for such matrices:
he showed that if an $n\times n$ matrix $A$ (over any field)
is a Kronecker product of $d\times d$ matrices $M_1,\dots,M_k$
(so $n=d^k$) $(d\ge 2)$ then changing only $n^{1+\eps}$ entries of
$A$ one can reduce its rank to $\le n^{1-\gamma}$,
where $1/\gamma$ is roughly $2^d/\eps^2$.

In this note we improve this result in two directions.  First,
we do not require the matrices $M_i$ to have equal size.
Second, we reduce $1/\gamma$ from exponential in $d$
to roughly $d^{3/2}/\eps^2$ (where $d$ is the maximum
size of the matrices $M_i$), and to nearly linear (roughly $d/\eps^2$)
for matrices $M_i$ of sizes within a constant factor of each other.

As an application of our results we significantly expand the class
 of Hadamard matrices that are known not to be Valiant-rigid;
 these now include the Kronecker products of Paley-Hadamard matrices
 and Hadamard matrices of bounded size. 
\end{abstract}

\section{Introduction}

\subsection{Recent upper bounds on rigidity}

In his celebrated MFCS'77 paper~\cite{valiant}, Leslie Valiant introduced
the notion of matrix rigidity as a tool to prove lower bounds for arithmetic circuits. Since then, several other important problems in
complexity theory have been reduced to proving rigidity lower bounds for explicit families of  matrices (see, e.g., \cite{razborov, goldreich-oded} and the survey~\cite{survey}).  
\begin{definition} Let $\fff$ be a field. For a matrix $A\in \fff^{n\times m}$
and a target rank $0\leq r\leq \min(n,m)$ let $R_{\fff}(A, r)$
denote the 
minimum    
number of non-zero entries in a matrix $Z\in \fff^{n\times m}$ such that $\rank(A-Z)\leq r$. The function $R_{\fff}(A, \cdot)$ is called the
\emph{rigidity}   
of $A$ over $\fff$.
\end{definition}

Valiant~\cite{valiant} proved that if for some $\eps>0$ the sequence of matrices $A_n\in \fff^{n\times n}$ satisfies 
\begin{equation}\label{eq:Valiant-rigid}
 R_{\fff}(A_n, n/\log\log n)\geq n^{1+\eps}, 
\end{equation}
then the linear functions $x\mapsto A_n x$ cannot be computed by arithmetic circuits of size $O(n)$ and depth $O(\log n)$.
Following \cite{dvir-edelman},   
we say that a family of matrices $A_n$ is \emph{Valiant-rigid} if it satisfies Eq.~\eqref{eq:Valiant-rigid} for some $\eps >0$ and all sufficiently large $n$. By saying that a family $\mathcal{F}$ of matrices  in not Valiant-rigid we mean that none of the subsequences of matrices of $\mathcal{F}$ of increasing order is Valiant-rigid. 

The problem of constructing explicit Valiant-rigid matrices has been attacked for more than four decades, but still remains open (see the survey~\cite{survey}). Over this time a few candidates of Valiant-rigid families were proposed that included Hadamard matrices~\cite{pudlak-savicky, razborov}, circulants~\cite{codenotti2000some}, Discrete Fourier Transorm (DFT) matrices~\cite{valiant}, incidence matrices of projective planes over a finite field~\cite{valiant}.    

In 2017, Alman and Williams~\cite{alman-williams} proved, that,
contrary to expectations,  
the Walsh--Hadamard matrices are not Valiant-rigid over $\qqq$. 
Subsequently, most of other long-studied candidates for rigidity
were shown 
not to be  
Valiant-rigid. 
Dvir and Edelman~\cite{dvir-edelman} proved that 
$G$-circulants\footnote{For a finite abelian group $G$, a
\emph{$G$-circulant} is a $|G|\times |G|$ matrix of the form
$M(f)$ with entries $M(f)_{xy}=f(x-y)$ $(x,y\in G)$
where $f$ is any function with domain $G$.}
are not Valiant-rigid over $\fff_p$ for $G$ the additive group of $\fff_p^n$. 
Dvir and Liu~\cite{dvir-liu} proved that DFT matrices, circulant matrices, and more generally, $G$-circulant matrices for any abelian group $G$,  
are not Valiant-rigid over $\ccc$.
Moreover, as observed in~\cite{field-matters}, the results of Dvir and Liu imply
that the Paley-Hadamard matrices and the Vandemonde matrices with
a geometric progression as generators are not Valiant-rigid over
$\ccc$, and the incidence matrices of projective planes over finite fields
are not Valiant-rigid either over $\fff_2$ (contrary to Valiant's
suggestion~\cite{valiant}) or over $\ccc$.

Upper bounds on the rigidity of 
the Kronecker powers of a fixed matrix play an important role in these results. Indeed, the Walsh--Hadamard matrices are just the Kronecker powers of $H_2 = \left(\begin{matrix}
1 & 1\\
1 & -1
\end{matrix}\right)$. 
Furthermore, in order to show that DFT and circulant matrices are not Valiant-rigid, as the first step, Dvir and Liu~\cite{dvir-liu} prove that
generalized Walsh--Hadamard matrices (Kronecker powers of a DFT matrix)
are not Valiant-rigid. 

Hence,
one may expect that strong
upper bounds on the rigidity of  
Kronecker products
of small matrices  
may lead to
upper bounds on the rigidity of  
interesting new    
families of matrices, and so will further contribute to our intuition of where not to look for Valiant-rigid matrices.

Recently, Josh Alman~\cite{alman} proved that Kronecker products of
square matrices of any \emph{fixed} size  
are not Valiant-rigid. More precisely, he proved the following result. 

\begin{theorem}[Alman]\label{thm:Alman}
Given $d\ge 2$ and $\eps > 0$, there exists
$\gamma = \Omega\left(\dfrac{d\log d}{2^d}\cdot \dfrac{\eps^2 }{\log^2(1/\eps)}\right)$ such that the following holds for any sequence of matrices
$M_1, M_2, \ldots, M_k \in \fff^{d\times d}$.
Let $M = \bigotimes\limits_{i=1}^{k} M_i$ and $n=d^k$.
Then
 $R_{\fff}(M, n^{1-\gamma})\leq n^{1+\eps}$. 
\end{theorem}

\subsection{Our results: improved bounds, non-uniform sizes}

In this note we improve Alman's result in two directions.
First, we do not require the matrices $M_i$ to have equal size.
Second, we reduce $1/\gamma$ from exponential in $d$
to roughly $d^{3/2}/\eps^2$ (where $d$ is the maximum
size of the matrices), and to nearly linear (roughly $d/\eps$)
for matrices $M_i$ of sizes within a constant
factor of each other.

This means that for matrices of equal size we get a meaningful reduction of the
rank already for $k= \wto_{\eps}(d)$,
in contrast to Alman's result that kicks in when
$k$ reaches about $2^d$.  

\begin{theorem}\label{thm:unequal-intr-no-assum}
Given $d\ge 2$ and $\eps > 0$, there exists
$\gamma = \Omega\left(\dfrac{1}{d^{3/2}\log^3(d)}\cdot\dfrac{\eps^2}{\log^2(1/\eps)}\right)$ such that the following holds for any sequence of matrices
$M_1, M_2, \ldots, M_k$ where $M_i \in \fff^{d_i\times d_i}$
for some $d_i\le d$.
Let $M = \bigotimes\limits_{i=1}^{k} M_i$ and
$n = \prod\limits_{i=1}^{k} d_i$.
If $n\geq d^{1/\gamma}$, then
$R_{\fff}(M, n^{1-\gamma})\leq n^{1+\eps}$. 
\end{theorem}

If the $d_i$ are within a constant factor of each other, we
obtain the following stronger result.

\begin{theorem}  \label{thm:unequal-intr}  \label{thm:main-intr}
Given $d\ge 2$, $\eps > 0$, and a constant $c > 0$, there exists
$\gamma = \Omega_c\left(\dfrac{1}{d\log d}\cdot\dfrac{\eps^2}{\log^{2}(1/\eps)}\right)$ 
such that the following holds for any sequence of matrices
$M_1, M_2, \ldots, M_k$ where $M_i \in \fff^{d_i\times d_i}$
and $cd \le d_i\le d$.
Let $M = \bigotimes\limits_{i=1}^{k} M_i$ and
$n = \prod\limits_{i=1}^{k} d_i$.
If $n\geq d^{1/\gamma}$, then
$R_{\fff}(M, n^{1-\gamma})\leq n^{1+\eps}$. 
\end{theorem}

In the above theorem, the dependence of $\gamma$ on $c$ is nearly linear and the explicit formula can be found in Corollary~\ref{cor:linear-gap}.

The key strength of our results is that we do not need to assume
uniform size of the 
matrices participating in the Kronecker product.
Previous approaches depended on uniform size because of their reliance
either on a polynomial method, or on an induction on the size of the matrix. As an applications of the bound for matrices of non-uniform size we show that our results significantly expand the class of Hadamard matrices that are
 known not to be Valiant-rigid. We show that the Kronecker
 products of Paley-Hadamard matrices and Hadamard matrices of bounded size are not Valiant-rigid (see Sec.~\ref{sec:Hadamard-intr}).

Another strength of our improvement is that
our bounds are sufficiently strong to be fed into the machinery developed
by Dvir and Liu~\cite{dvir-liu} for matrices of ``well-factorable'' size.
Hence, we expect that this improvement might
lead to further applications. 

We note, that our
bound on $\gamma$ in Theorem~\ref{thm:unequal-intr} 
matches the Dvir--Liu bounds 
for Kronecker powers of specific classes of matrices, such as
the generalized Walsh--Hadamard matrices and
DFT matrices\footnote{The DFT (Discrete Fourier Transform)
matrix of a finite abelian group $G$ is
the character table of $G$.}  
of direct products of small abelian groups.

We also note that our upper bounds, similarly to upper bounds
in recent work \cite{alman-williams, dvir-edelman, dvir-liu, alman},
apply to a stronger notion of rigidity, called
\emph{row-column rigidity}.

\begin{definition}[Row-column rigidity] For a matrix $A\in \fff^{n\times n}$ and a target rank $0\leq r\leq n$, let $R^{rc}_{\fff}(A, r)$ be the minimal $t$ for which there exists $Z\in \fff^{n\times n}$ such that $\rank(A-Z)\leq r$ and every row and column of $Z$ has at most $t$ non-zero entries.
\end{definition}

In Theorems~\ref{thm:Alman},~\ref{thm:unequal-intr-no-assum},
and~\ref{thm:unequal-intr} the conclusion can be replaced
by   
$R^{rc}_{\fff}(M, n^{1-\gamma})\leq n^{\eps}$. 
Clearly, the latter statement is stronger, as for any
$A\in \fff^{n\times n}$ the inequality
$R_{\fff}(A, r)\leq n\cdot R_{\fff}^{rc}(A, r)$ holds. The row-column versions of Theorems~\ref{thm:unequal-intr-no-assum}, 
and~\ref{thm:unequal-intr} are stated and proved as Theorem~\ref{thm:row-column-main-noassum} and Corollary~\ref{cor:dtod2-bound}.

\subsection{The field}

We should point out that    Theorems~\ref{thm:Alman}-\ref{thm:unequal-intr}  make no  
assumption about the field $\fff$, and they use
elements of $\fff$ for the rank reduction.
This is in contrast to the results of Dvir and Liu
who require a field extension to achieve their
rank reduction.

If the field $\fff$ is the field of definition of the matrix $A$,
in \cite{field-matters} we call the corresponding rigidity function
$R_{\fff}(A,.)$ the \emph{strict rigidity} of $A$.  If $\fffbar$
denotes the algebraic closure of $\fff$ then we call $R_{\fffbar}(A,.)$
the \emph{absolute rigidity} of $A$ because, as shown in \cite{field-matters},
this gives the smallest possible rigidity among all extension fields.
A gap between these two quantities is
demonstrated  
in \cite{field-matters}.
Note that  
an upper bound on strict rigidity is a stronger statement than
the same upper bound on absolute rigidity.  

In this terminology,
Dvir and Liu give upper bounds on absolute rigidity, whereas
Alman's result and our results give upper bounds on strict rigidity.

\subsection{Application of our results: rigidity upper bounds for Hadamard matrices}\label{sec:Hadamard-intr}

We recall that an \emph{Hadamard matrix} is a square matrix whose entries are $+1$ and $-1$ and whose rows are mutually orthogonal. In addition to many other interesting properties, it was long believed that one can find a family of Valiant-rigid matrices among Hadamard matrices. 

Contrary to expectations, two of the most well-studied families of Hadamard matrices were recently shown to be not Valiant-rigid. In 2017, Alman and Williams~\cite{alman-williams} proved that Walsh-Hadamard matrices are not strictly Valiant-rigid. As a corollary to~\cite{dvir-liu}, in~\cite{field-matters} it was shown that Paley-Hadamard matrices are not absolutely Valiant-rigid. These results inspired the following conjecture.

\begin{conjecture}[Babai]
The family of known Hadamard matrices is not strictly Valiant-rigid.
\end{conjecture}

We mention that in addition to infinite families, new classes of Hadamard matrices arise as Kronecker products of a
steadily growing starter set of small
Hadamard matrices  with other known Hadamard matrices (see, e.g., surveys~\cite{hadamard-survey2, hadamard-survey1}). Indeed, note that if $H_1$ and $H_2$ are Hadamard matrices, then $H_1\otimes H_2$ is an Hadamard matrix as well.

As an application of our results for Kronecker products of matrices of non-uniform size we further expand the family of Hadamard matrices that are known not to be Valiant-rigid. 

\begin{theorem}\label{thm:main-Hadamard}
Let $\mathcal{F}_0$ be the family of Paley-Hadamard matrices and Hadamard matrices of bounded size. Let $\mathcal{F}$ be the family of all matrices that can be obtained as Kronecker products of some matrices from $\mathcal{F}_0$. Then $\mathcal{F}$ is not absolutely Valiant-rigid.
\end{theorem}

\begin{remark}
We note that Hadamard matrices are naturally defined over $\qqq$, while the theorem above shows that matrices in $\mathcal{F}$ are not sufficiently rigid when we make changes from $\ccc$. It is still open whether Paley-Hadamard matrices are strictly Valiant-rigid. If one proves that Paley-Hadamard matrices are not strictly Valiant-rigid, our proof will immediately yield the stronger version of the theorem above, that $\mathcal{F}$ is not strictly Valiant-rigid. 
\end{remark}
\begin{remark}
We note that Theorem~\ref{thm:main-Hadamard} does not follow from Alman's original upper bound for rigidity of Kronecker products (Theorem~\ref{thm:Alman}).
\end{remark}

A more general version of Theorem~\ref{thm:main-Hadamard} can be stated for Kronecker products of matrices of bounded size and matrices that are sufficiently not rigid. 

\begin{theorem}\label{thm:bounded-and-non-rigid-intr} Let $0<\varepsilon <1/2$ and $b\geq 2$. Let $\mathcal{F}$ be a family of matrices over $\fff$, such that for every $d\times d$ matrix $A\in \mathcal{F}$ either $d\leq b$, or 
\begin{equation}
R_{\fff}^{rc}\left(A, d^{1-\gamma}\right) \leq d^{\varepsilon} \quad \text{for} \quad \gamma = \dfrac{12 (\log\log d)^2}{\varepsilon^3 \log d}.
\end{equation} 
Then, for every sequence of matrices $M_1, M_2, \ldots M_k \in\mathcal{F}$, the $n\times n$ matrix $M = \bigotimes\limits_{i\in [k]} M_i$ either satisfies
$
R_{\fff}^{rc}\left(M, {n}/{\log n}\right)\leq n^{6\varepsilon},
$  
 or $n$ is bounded above by a function of $b$ and $\varepsilon$.
\end{theorem}
\begin{corollary}
Let $\mathcal{F}$ be a family of square matrices over $\fff$. If for every $0<\varepsilon<1/2$ there exists $b\geq 2$ such that $\mathcal{F}$ satisfies the assumptions of Theorem~\ref{thm:bounded-and-non-rigid-intr}, then the family of Kronecker products of matrices from $\mathcal{F}$ is not Valiant-rigid over $\fff$.
\end{corollary}

We prove Theorems~\ref{thm:main-Hadamard} and~\ref{thm:bounded-and-non-rigid-intr}  in Section~\ref{sec:Hadamard}.

\subsection{Our approach}

In order to prove Theorem~\ref{thm:Alman}, Alman~\cite{alman} first uses a beautiful trick to deduce the claim  for Kronecker products of $2\times 2$ matrices.
He observes that it is sufficient to prove the claim for $R =\left( \begin{matrix} 1 & 1\\
1 & 0
\end{matrix}\right)
$. The Kronecker powers of $R$ have low rigidity since they are very sparse. 
After that, he applies induction on the size $d$ of the matrices involved
in the Kronecker product.
The technically involved induction argument leads to the
factor $2^{-d}$ in~$\gamma$.    

Our proof is simpler and omits
induction. Instead, we observe that an idea,
somewhat similar to Alman's proof of the base case
$d=2$, can be applied for any~$d$.  

Our key observation is that any $d\times d$ matrix can be written as a product
of at most $2d$ very sparse matrices and $2d$ permutation matrices.

Specifically,   
for a vector $x\in \fff^d$ define a $d\times d$ matrix
\begin{equation}
G_d(x)  = \left(\begin{matrix}
		I_{d-1} \\
		0 
		\end{matrix}\  x\right) = \left(\begin{matrix}
		1 & 0 & 0 \ldots & 0 & x_1\\
		0 & 1 & 0 \ldots & 0 & x_2\\
		\ldots \\
		0 & 0 & 0 \ldots & 1 & x_{d-1} \\
		0 & 0 & 0 \ldots & 0 & x_d
		\end{matrix}\right).
\end{equation} 
We are going to call the matrices of this form  the \textit{V-matrices} for the pattern of their non-zero entries. Next, it is not hard to verify that any $d\times d$ matrix can be written as
\[ A = P_1 \cdot  G_d(y)^T  \cdot \left(\begin{matrix}
		B & 0\\
		0 & \lambda
		\end{matrix}\right)   \cdot G_d(x)  \cdot P_2,
\]
where $\lambda \in \{0, 1\}\subseteq \fff$, $B\in\fff^{ (d-1)\times (d-1)}$ and $P_1, P_2$ are permutation matrices. By repeating this procedure for $B$ at most $d-2$ times one ends up with a product of $2d-2$ V-matrices, a diagonal matrix, and permutation matrices (see Section~\ref{sec:structure}).

As was observed
in \cite{dvir-liu},  
the row-column rigidity of the product can be controlled by the
row-column rigidity of each component. 

\begin{lemma}\label{lem:rigidity-prod} For arbitrary $d\times d$ matrices $A$ and $B$ over a field $\fff$
\begin{equation}
R^{rc}_{\fff}(A\cdot B, r+s)\leq R^{rc}_{\fff}(A, r)\cdot R^{rc}_{\fff}(B, s)
  \,.  
\end{equation}
\end{lemma}

Recall that a \emph{monomial matrix} is a matrix where each row and each
column has at most one non-zero element.  A matrix is monomial
exactly if it is the product of a diagonal matrix and a permutation
matrix.  
A Kronecker product of
monomial  
matrices is itself a
monomial 
matrix. Moreover, if $P$ is a monomial matrix then 
$R^{rc}_{\fff}(P, 0) = 1$.   
Thus, in order to prove Theorem~\ref{thm:main-intr} one only needs to
show strong 
upper bounds on row-column rigidity  
for Kronecker products of V-matrices.

To bound the row-column rigidity of a matrix $M$ that is a Kronecker product of V-matrices one needs
to observe that 
most of the non-zero entries are concentrated in just a few columns and rows,
and so these entries form a low-rank matrix. 

We discuss our strategy
for upper bounds on rigidity  
of 
Kronecker products in Section~\ref{sec:strategy}.
We prove Theorem~\ref{thm:main-intr}
for the case of matrices of equal size 
in Section~\ref{sec:equal-size}. We discuss rigidity bounds for
matrices of unequal sizes in Section~\ref{sec:nonequal-size}.
Some standard tail bounds for binomial distributions are reviewed
in Section~\ref{sec:tail-bounds}. We expand the family of Hadamard matrices known not to be Valiant rigid in Section~\ref{sec:Hadamard}. The proof of Obs.~\ref{obs:factor} omitted in Section~\ref{sec:strategy} is provided in Appendix~\ref{sec:proof-obs-factor}.

\subsection{Notation}

We use $[n]$ to denote the set $\{1, 2, \ldots , n\}$. $\fff$ denotes a field. Throughout the paper we use $\vec{d}$ to denote the vector $(d_1, d_2, \ldots d_k)$, where each $d_i\geq 2$. We define $[\vec{d}] = [d_1]\times [d_2] \times \ldots \times [d_k]$. We say that $X \in \fff^{\vec{d}}$ if $X = (X_1, X_2, \ldots, X_k)$, where $X_i \in \fff^{d_i}$.

We use standard asymptotic notation.  
Let $f(n), g(n)$ be non-negative functions. We say that $f(n)= O(g(n))$ if there exists a constant $C$ such that $f(n)\leq Cg(n)$ for all
sufficiently large $n$.   
We say that $f(n) = \Omega(g(n))$ if $g(n) = O(f(n))$.
Finally, $f(n) = \Theta(g(n))$ if $f(n) = O(g(n))$ and $f(n) = \Omega(g(n))$.

\section*{Acknowledgments}
 The author is grateful to his advisor 
L\'aszl\'o Babai for helpful discussions, his 
help in improving the organization of the paper, and
for pointing out improvements to the results and simplifications
        of the proofs. The author is partially supported by Prof. L\'aszl\'o Babai's
   NSF Grant CCF 1718902.  All statements made in this paper reflect the
   author's views and have not been evaluated or endorsed by
   the NSF.

\section{Tail bounds}
\label{sec:tail-bounds}

In this section we review a
classical concentration inequality which we apply to
the binomial distribution.

\begin{theorem}[Bernstein inequality]\label{upper-bern}
Let $L\geq 0$ and $\delta>0$.  Let
$X_1,\dots,X_n$ be   
independent random variables satisfying $|X_i - \expv(X_i)|\leq L$ for $i\in [n]$. Let $X = \sum\limits_{i=1}^{n} X_i$. Then we have
\[ \prob(X\geq \expv(X)+\delta)\leq \exp\left(-\dfrac{\delta^2/2}{\sum\limits_{i=1}^{n}\Var(X_i)+L\delta/3}\right). \] 
\end{theorem}

\begin{corollary}\label{cor:bernstein} Let $d>1$. Let $X_i$ be
i.\,i.\,d.   
$\{0,1\}$-valued 
random variables such that $\prob(X_i = 1) = 1-\prob(X_i = 0) = 1/d$
for $i\in [n]$. Let $X = \sum\limits_{i=1}^{n} X_i$.  Let $\delta<1/d$. Then
\[ \prob\left[X\geq \left(\dfrac{1}{d}+\delta\right)n\right]\leq 
\exp\left(-\dfrac{\delta^2 n}{3d}\right) \quad \text{and}\quad 
\prob\left[X\leq \left(\dfrac{1}{d}-\delta\right)n\right]\leq 
\exp\left(-\dfrac{\delta^2 n}{3d}\right). \]
\end{corollary}
\begin{proof} Note that $\expv(X_i) = 1/d$ and $\Var(X_i) = (d-1)/d^2$. 
The first bound follows from Theorem~\ref{upper-bern}
with 
$L = (d-1)/d$. The second bound follows from Theorem~\ref{upper-bern} 
with $L = (d-1)/d$ by applying it to $Y_i = \expv(X_i) - X_i$.
\end{proof}

We also use the following
elementary 
bound.
\begin{lemma}\label{lem:tail-bound} Let $0\leq k\leq n$. Then 
$
\sum\limits_{i=0}^{k} \binom{n}{i} \leq \left(\dfrac{\eee n}{k}\right)^{k}.
$
\end{lemma}

\section{A strategy for upper bounds on the rigidity of Kronecker products}
\label{sec:strategy}

\subsection{Expressing a matrix as a product of V-matrices}\label{sec:structure}

Recall that for a vector $x\in \fff^d$ we define a $d\times d$ V-matrix
\begin{equation}
G_d(x)  = \left(\begin{matrix}
		I_{d-1} \\
		0 
		\end{matrix}\  x\right) = \left(\begin{matrix}
		1 & 0 & 0 \ldots & 0 & x_1\\
		0 & 1 & 0 \ldots & 0 & x_2\\
		\ldots \\
		0 & 0 & 0 \ldots & 1 & x_{d-1} \\
		0 & 0 & 0 \ldots & 0 & x_d
		\end{matrix}\right).
\end{equation}
Let $\vec{d} = (d_1, d_2, \ldots, d_k)$. For $X = (X_1, X_2, \ldots X_k)$ with $X_i\in \fff^{d_i}$ for $i\in [k]$, define
\begin{equation}
G_{\vec{d}\, }(X) = \bigotimes_{i=1}^{k} G_{d_i}(X_i).
\end{equation}
In the case when $\vec{d}$ consists of $k$ equal coordinates $d_i = d$, we use notation $G_{d, k}(X) := G_{\vec{d}\, }(X)$.

\begin{definition} A square matrix $P$ is called a \textit{permutation matrix}, if every row and every column of $P$ has precisely one non-zero entry, which is equal to 1.  
\end{definition}

Recall, that arbitrary permutation of columns (rows) of a matrix can be represented by right (left, respectively) multiplication by a permutation matrix. 

The multiplication of an $m\times d$ matrix $A$ by $G_d(x)$ from the right corresponds to keeping the first $d-1$ columns of $A$ unchanged and replacing the last column of $A$ with a linear combination of the columns of $A$ with the coefficients given by $x$. We make the following observation. 

\begin{observation}\label{obs:factor} For any matrix $A \in \fff^{d\times d}$ there exist $B\in \fff^{(d-1)\times (d-1)}$, vectors $x, y\in \mathbb{F}^d$, $\lambda\in \{0, 1\}\subseteq \fff$ and permutation matrices $P_1, P_2\in \fff^{d\times d}$ such that
\begin{equation}\label{eq:factor} A = P_1 \cdot  G_d(y)^T  \cdot \left(\begin{matrix}
		B & 0\\
		0 & \lambda
		\end{matrix}\right)   \cdot G_d(x)  \cdot P_2.
\end{equation}
\end{observation}
\begin{proof} See Appendix~\ref{sec:proof-obs-factor}.
\end{proof}

\begin{corollary}\label{cor:decomp} For any matrix $A\in \fff^{d\times d}$ there exist a diagonal matrix $W$, $2(d-1)$ vectors $X_i, Y_i$ for $i\in [d-1]$, and $2(d-1)$ permutation matrices $P_i, Q_i$ for $i\in [d-1]$ such that 
\begin{equation}
A = Q_{1}\cdot G_{d}(Y_{1})^T\cdot Q_{2}\cdot\ldots \cdot G_{d}(Y_{d-1})^T\cdot W \cdot G_d(X_{d-1})\cdot P_{d-1}\cdot \ldots \cdot G_d(X_{1})\cdot P_{1}.
\end{equation}
\end{corollary}
\begin{proof} Let $t\in [d-1]$. Observe that for $x \in \fff^{t}$ and the vector $y \in \fff^{d}$, whose first $d-t$ coordinates are 0, and last $t$ coordinates are equal to the corresponding coordinates of $x$,
\begin{equation}
\left(\begin{matrix}
G_{d-t}(x) & 0\\
0 & I_{t}
\end{matrix}\right)  = P_1 G_{d}(y) P_2,
\end{equation} 
where $I_{t}$ is the $t\times t$ identity matrix and $P_1, P_2$ are the permutation matrices that exchange the first $d-t$ and the last $t$ rows and columns, respectively.

Then, the claim of the corollary follows from Observation~\ref{obs:factor} by applying it recursively  $d-1$ times to the remaining top-left block $B$, until we are left with a diagonal matrix.
\end{proof}

By combining the above corollary with Lemma~\ref{lem:rigidity-prod} we
obtain   
the following
inequality.  

\begin{lemma}\label{lem:rigid-reduct-Gd} Let $d\geq 2$, $r\leq d^k$ and $M_1, M_2, \ldots, M_k \in \fff^{d\times d}$. Then
\begin{equation}
R_{\fff}^{rc}\left(\bigotimes\limits_{i=1}^{k} M_i, (2d-2)r\right) \leq \left(\max\limits_{X\in \fff^{d\times k}}R_{\fff}^{rc}\left(G_{d, k}(X), r\right)\right)^{2d-2}.  
\end{equation}
\end{lemma}
\begin{proof} By Corollary~\ref{cor:decomp}, for each $i\in [k]$ there exist $2(d-1)$ vectors $X_{i}^{(j)}, Y_{i}^{(j)}\in \fff^{d}$, $2(d-1)$ permutation matrices $P_{i}^{(j)}, Q_{i}^{(j)}$ and a diagonal matrix $W_i$ such that
\[ M_i = \prod\limits_{j=1}^{d-1}\left(Q_{i}^{(j)}\cdot G_{d}(Y^{(j)}_i)^T\right)\cdot  W_i \cdot \prod\limits_{j=1}^{d-1}\left(G_d(X^{(d-j)}_i)\cdot P^{(d-j)}_i\right). \]  
Then, using that $(A_1\otimes B_1)\cdot (A_2\otimes B_2) = (A_1A_2)\otimes (B_1B_2)$ holds for any matrices, we get 
\[\bigotimes\limits_{i=1}^{k} M_i = \prod\limits_{j=1}^{d-1} \left(\bigotimes\limits_{i=1}^{k} Q_{i}^{(j)} \bigotimes\limits_{i=1}^{k} G_{d}(Y^{(j)}_i)^T\right)\cdot \left(\bigotimes\limits_{i=1}^{k} W_i\right)\cdot \prod\limits_{j=1}^{d-1} \left( \bigotimes\limits_{i=1}^{k} G_{d}(X^{(d-j)}_i) \bigotimes\limits_{i=1}^{k} P_{i}^{(d-j)}\right). 
\]
Let $X^{(j)}\in \fff^{d\times k}$ be a matrix whose $i$-th column is $X_{i}^{(j)}$. Define $Y^{(j)}\in \fff^{d\times k}$ similarly. Let
\begin{equation} Q^{(j)} = \bigotimes\limits_{i=1}^{k} Q_{i}^{(j)}, \qquad P^{(j)} = \bigotimes\limits_{i=1}^{k} P_{i}^{(j)},\qquad W = \bigotimes\limits_{i=1}^{k} W_i.
\end{equation} 
Then, $Q^{(j)}$ and $P^{(j)}$ are permutation matrices, and $W$ is a diagonal matrix, and
\begin{equation}
M = \bigotimes\limits_{i=1}^{k} M_i = \prod\limits_{j=1}^{d-1}\left(Q^{(j)}\cdot G_{d, k}(Y^{(j)})^T\right)\cdot  W \cdot \prod\limits_{j=1}^{d-1}\left(G_{d, k}(X^{(d-j)})\cdot P^{(d-j)}\right).
\end{equation}
Thus, by Lemma~\ref{lem:rigidity-prod},
\begin{equation}
\begin{split}
R_{\fff}^{rc}\left(M, (2d-2)r\right) & \leq \prod\limits_{j=1}^{d-1} R_{\fff}^{rc}\left(G_{d, k}(Y^{(j)}), r\right)\cdot \prod\limits_{j=1}^{d-1} R_{\fff}^{rc}\left(G_{d, k}(X^{(j)}), r\right) \leq \\
& \leq  \left(\max\limits_{X\in \fff^{d\times k}}R_{\fff}^{rc}\left(G_{d, k}(X), r\right)\right)^{2d-2}.   \qedhere   
\end{split}       
\end{equation} 
\end{proof}

Furthermore, we note that a similar statement holds for Kronecker products of matrices of not necessarily equal size.

\begin{lemma}\label{lem:rigid-reduct-Gd-gen} Let $d_i\geq 2$ and $M_i\in \fff^{d_i\times d_i}$ for $i\in [k]$. Assume $d_k\geq d_i$ for $i\in [k]$. Then
\begin{equation}
R_{\fff}^{rc}\left(\bigotimes\limits_{i=1}^{k} M_i, (2d_k-2)r\right) \leq \left(\max\limits_{X\in \fff^{\vec{d}}}R_{\fff}^{rc}\left(G_{\vec{d}\, }(X), r\right)\right)^{2d_k-2}.  
\end{equation}
\end{lemma}
\begin{proof} For $M_i$ with $d_i<d_k$ we multiply the decomposition given by Corollary~\ref{cor:decomp}, by $2(d_k - d_i)$ identity matrices from the left and from the right, so that every $M_i$ is decomposed into the product of precisely $4d_k - 3$ matrices. After this, the proof is identical to the proof of the previous lemma. 
\end{proof}

\subsection{General approach to upper bounds for the rigidity of {$G_{\vec{d}\, }(X)$}}
\label{sec:rigid-Gd-gen}

As we see from Lemmas~\ref{lem:rigid-reduct-Gd} and~\ref{lem:rigid-reduct-Gd-gen}, in order to prove Theorems~\ref{thm:unequal-intr-no-assum} and~\ref{thm:main-intr} it is sufficient to bound the row-column rigidity of Kronecker products of V-matrices. We describe a general approach to such bounds.

The key observation is that most of the non-zero entries of $G_{\vec{d}\, }(X)$ are concentrated in just a few columns and a few rows. Hence, after deleting these columns and rows we expect to get a matrix with very sparse rows and columns.

Recall, that for $\vec{d} = (d_1, d_2, \ldots, d_k)$ we define $[\vec{d}] = [d_1]\times [d_2]\times \ldots \times [d_k]$. The rows and the columns of $G_{\vec{d}\, }(X)$ are indexed by tuples in $[\vec{d}]$.

For a $d_i\times d_i$ V-matrix all rows, except the $d_i$-th row, have up to two non-zero entries, while the $d_i$-th row has only one non-zero entry. For V-matrices, all columns, except the last one, have one non-zero entry and the last column may have up to $d_i$ non-zero entries. Therefore, the densest rows are indexed by tuples with a small number of $i$-th coordinates being equal $d_i$ and the densest columns are indexed by tuples with a large number of $i$-th coordinates being equal $d_i$, for $i\in [k]$.

Let $w: \rrr\rightarrow \rrr$ be a function. Define a score of a string $x\in [\vec{d}]$ as 
\begin{equation}\label{eq:sx}
s(x) = \sum\limits_{i=1}^{k} w(d_i)\mathbf{1}[x_i = d_i].
\end{equation}
Consider
the  
uniform distribution on $[\vec{d}]$. Then $\chi(i) = w(d_i)\mathbf{1}[x_i = d_i]$ are independent random variables, and $s(x) = \sum\limits_{i=1}^{k} \chi(i)$ with 
\begin{equation}
 \expv_x [s(x)] = \sum\limits_{i=1}^{k} \frac{w(d_i)}{d_i} 
\quad \text{and} \quad \Var[s(x)] = \sum\limits_{i=1}^{k} w(d_i)^{2}\left(\dfrac{d_i - 1}{d_i^2}\right). 
\end{equation} 
For $\delta>0$ we define a pair of sets of strings of high and low scores, respectively
\begin{equation}
\mathcal{C}(\delta) = \{x\in [\vec{d}] \mid s(x) \geq \expv_y [s(y)]+\delta\} \quad \text{and} \quad \mathcal{R}(\delta) = \{x\in [\vec{d}] \mid s(x) \leq \expv_y [s(y)]-\delta\}. 
\end{equation}
These sets correspond to indices of the most dense columns and rows, respectively. Note, that one may use concentration inequalities to bound the sizes of
these sets (see Sec.~\ref{sec:tail-bounds}). \\
The matrix $E$ defined by the union of the rows  of $G_{\vec{d}}(X)$ in $\mathcal{R}(\delta)$ and the columns of $G_{\vec{d}}(X)$ in $\mathcal{C}(\delta)$ has low rank. Now we want to count the number of non-zero entries in rows and columns of $G_{\vec{d}}(X) - E$. Define
\begin{equation}
T_c(\delta, y) = \{x\in [\vec{d}]\setminus \mathcal{R}(\delta) \mid \forall i:\  y_i \in \{x_i, d_i\}\}\quad \text{and} \quad  M_c(\delta) = \max\limits_{y\in [\vec{d}]\setminus \mathcal{C}(\delta)} |T_c(\delta, y)|;
\end{equation}
\begin{equation}
T_r(\delta, x) = \{ y\in [\vec{d}]\setminus \mathcal{C}(\delta) \mid\forall i:\   y_i \in \{x_i, d_i\} \} \quad \text{and} \quad  M_r(\delta) = \max\limits_{x\in [\vec{d}]\setminus \mathcal{R}(\delta)} |T_r(\delta, x)|.
\end{equation}

With this notation, we have the following inequality.  

\begin{lemma}\label{lem:rig-general} Let $\delta>0$. Then, for any $X\in \fff^{\vec{d}}$,
\[ R_{\fff}^{rc}(G_{\vec{d}\, }(X), |\mathcal{C}(\delta)|+|\mathcal{R}(\delta)|)\leq \max(M_c(\delta), M_r(\delta)).\] 
\end{lemma}
\begin{proof} Let $E$ be the matrix obtained from $G_{\vec{d\, }}(X)$ by changing to 0 every entry that is not in a column with index in $\mathcal{C}(\delta)$ and is not in a row with index in $\mathcal{R}(\delta)$. Then 
\begin{equation}
 \rank(E) \leq |\mathcal{C}(\delta)|+|\mathcal{R}(\delta)|.
 \end{equation}
Let $Z = G_{\vec{d}\, }(X) - E$ and $x, y \in [\vec{d}]$. Observe that the  entry of $G_{\vec{d}\, }(X)$ with coordinates $(x, y)$ is non-zero only if for every $i\in [k]$ either $x_i = y_i$ or $y_i = d_i$. Therefore, every row of $Z$ has at most $M_{r}(\delta)$ non-zero entries, and every column of $Z$ has at most $M_{c}(\delta)$ non-zero entries.
\end{proof}

Hence, in order to prove
an upper bound on the rigidity of  
$G_{\vec{d}\, }(X)$,   
one just may come up with a good choice of weights $w$ and a threshold $\delta>0$ which make all the quantities $|\mathcal{C}(\delta)|$, $|\mathcal{R}(\delta)|$, $M_c(\delta)$ and  $M_r(\delta)$ small.

\section{Rigidity of Kronecker products of matrices of uniform size}\label{sec:equal-size}

Now, we show how the
approach, described above, provides  
a strong upper bound on the rigidity of  
Kronecker products of matrices of uniform size.
We follow the notation introduced earlier.
\begin{theorem}  \label{thm:Gdnonrigid}
Given $d\ge 2$ and $0<\eps <1$, there exists
$\gamma = \Theta\left(\dfrac{1}{d\log d}\cdot\dfrac{\eps^2}{\log^{2}(1/\eps)}\right)$ 
such that for all $k\geq 1$ and $X\in \fff^{d\times k}$ we have 
$
  R_{\fff}^{rc}(G_{d, k}(X), n^{1-\gamma}) \leq n^{\eps/d}
$, where $n = d^k$.
\end{theorem}

\begin{proof} We use the approach described in the previous section. We take $w(x) = 1$, then $s(x)$ given by Eq.~\eqref{eq:sx} counts the number of coordinates equal to $d$. A simple computation gives
\begin{equation}
\expv[s(x)] = k/d \qquad \text{and} \qquad \Var[s(x)] = k(d - 1)/d^2.
\end{equation}
Then, by the Bernstein inequality (see Corollary~\ref{cor:bernstein}), for $\delta<1/d$,
\begin{equation}
|\mathcal{C}(\delta k)| = d^k \cdot  \prob\left[s(x)\geq k/d+\delta k \right]\leq 
\exp \left(\left(\ln d - \frac{\delta^2d}{3}\right)k\right).
\end{equation}

\begin{equation}
|\mathcal{R}(\delta k)| = d^k \cdot  \prob\left[s(x)\leq k/d-\delta k \right] \leq  \exp \left(\left(\ln d - \frac{\delta^2d}{3}\right)k\right).
\end{equation}
Therefore, for $\delta = \Theta\left(\dfrac{\eps}{\log(1/\eps)d}\right)$,
\begin{equation}
 |\mathcal{R}(\delta k)|+|\mathcal{C}(\delta k)|\leq d^{(1-\gamma) k} = n^{1-\gamma},\quad \text{for some}\quad \gamma = \Theta\left(\dfrac{\eps^2}{d\log d\log^2(1/\eps)}\right). 
\end{equation}

For a string $x\in [d]^k$ define the set $S_x = \{i\mid x_i = d\}$. Recall, that for every $y \in [d]^k\setminus \mathcal{C}(\delta k)$ we have $|S_y|\leq k(1/d+\delta)$ and for every $x \in [d]^k\setminus \mathcal{R}(\delta k)$ we have $|S_x|\geq k(1/d-\delta)$.

For each $y \in [d]^k\setminus \mathcal{C}(\delta k)$, a vector $x\in T_{c}(\delta k, y)$ can be described by picking a subset $U$ of $S_y$ of the size at most $2\delta k$ such that $U = S_y\setminus S_x$ and by picking $x_j\in [d-1]$ for every $j\in U$. So, by Lemma~\ref{lem:tail-bound}, 
\begin{equation}
|T_{c}(\delta k, y)|  \leq \sum\limits_{i\leq 2\delta k} \binom{|S_y|}{i}(d-1)^i\leq \left(\dfrac{\eee|S_y|(d-1)}{2\delta k}\right)^{2\delta k}\leq \exp(2\delta\ln(3/\delta)k). 
\end{equation} 
Similarly, for every row index $x\in [d]^k\setminus \mathcal{R}(\delta k)$, a column $y\in T_{r}(\delta k, x)$ can be described by picking a subset $U\subseteq [k]\setminus S_x$ of size at most $2\delta k$ and by setting $y_j = d$ for $j\in U$ and $y_j = x_j$ for $j\notin U$. Hence, by Lemma~\ref{lem:tail-bound}, 
\begin{equation}
|T_{r}(\delta k, x)|  \leq \sum\limits_{i\leq 2\delta k} \binom{k-|S_x|}{i}\leq \left(\dfrac{\eee k}{2\delta k}\right)^{2\delta k}\leq \exp(2\delta\ln(2/\delta)k).
\end{equation} 
 Therefore, 
\begin{equation}
\max(M_c(\delta k), M_r(\delta k)) \leq \exp\left(2\delta\ln(3/\delta)k\right)\leq d^{\displaystyle \eps k/d}, \quad \text{for some } \delta = \Theta\left(\dfrac{\eps}{\log(1/\eps)d}\right).
\end{equation} 
Hence, the conclusion of the theorem follows from Lemma~\ref{lem:rig-general}.
\end{proof}

Finally, we can deduce our improved bound for the rigidity of
Kronecker products of matrices of uniform size.

\begin{theorem}\label{thm:equal-sizes}
Given $d\ge 2$ and $0<\eps <1$, there exists
$\gamma = \Omega\left(\dfrac{1}{d\log d}\cdot \dfrac{\eps^2}{\log^2(1/\eps)}\right)$
such that the following holds for any
$M_1, M_2, \ldots, M_k \in \fff^{d\times d}$ with $k > 1/\gamma$.
Let $M = \bigotimes\limits_{i=1}^{k} M_i$ and $n = d^k$.  Then
\[ R_{\fff}^{rc}\left(M, n^{1-\gamma}\right)\leq n^{\eps}, \quad \text{ and so } \quad R_{\fff}\left(M, n^{1-\gamma}\right)\leq n^{1+\eps}.\] 
\end{theorem}
\begin{proof} By Theorem~\ref{thm:Gdnonrigid} and Lemma~\ref{lem:rigid-reduct-Gd}, for $\eps' = \eps/2$ we have
\[ R_{\fff}^{rc}\left(M, 2d\cdot d^{(1-\gamma')k}\right)\leq d^{2\eps' k} \quad \text{for some} \quad \gamma' = \Theta\left(\dfrac{\eps^2}{\log^2(1/\eps)d\log d}\right).\]
Pick $\gamma = \gamma'/3$ and note that for $k>1/\gamma$ we have $2d\cdot d^{-2\gamma k}\leq 2d\cdot d^{-2}\leq 1$.
\end{proof}

\section{Rigidity bounds for matrices of unequal sizes} 
\label{sec:nonequal-size}

In this section we show how a similar approach may be used to
show upper bounds on 
the rigidity of Kronecker products of matrices of
not necessarily equal sizes.  

\begin{theorem}\label{thm:unequal} 
Let $0<\eps<1$ and let $w(x)\geq 1$ be a non-decreasing function. There exists an absolute constant $c>0$ such that the following is true. Let $2\leq d_1 \leq d_2 \leq \ldots \leq d_k$ be integers. Let $L>0$ and $K>0$ be such that  
\begin{enumerate}[(i)]
\item $\sum\limits_{i=1}^{k}w(d_i)/d_i\leq K\left(\sum\limits_{i=1}^{k} \log d_i / \log d_k\right) \cdot (w(d_1)/d_k)$, and  
\item $\sum\limits_{i=1}^{k}w(d_i)/d_i\geq L \left(\sum\limits_{i=1}^{k} \log d_i/\log d_k\right) \cdot (w(d_k)/d_k)$.
\end{enumerate}
 Let $n = \prod\limits_{i=1}^{k} d_i$. Consider arbitrary matrices $M_i\in \fff^{d_i\times d_i}$ and let $M = \bigotimes\limits_{i=1}^{k} M_i$. Then 
 \[ R^{rc}_{\fff}(M, 2d_k\cdot n^{1-\gamma})\leq n^{\eps}\quad \text{where}\quad \gamma= \dfrac{c L \eps^2}{d_k \log d_k \cdot K^2\log^2(K/\eps)}. \] 
\end{theorem}
\begin{proof}
 We follow the notation of Section~\ref{sec:rigid-Gd-gen}. Recall that
\begin{equation}
 m := \expv [s(x)] = \sum\limits_{i=1}^{k} w(d_i)/d_i \quad \text{and} \quad \Var[s(x)] = \sum\limits_{i=1}^{k} w(d_i)^{2}\left(\dfrac{d_i - 1}{d_i^2}\right)\leq w(d_k)m. 
\end{equation} 

Then, by the Bernstein inequality, for $0<\delta<1$,

\begin{equation}\label{eq:cdel-uneq}
\begin{split}
|\mathcal{C}(\delta m)| = n \cdot  \prob\left[s(x)\geq m+\delta m \right]\leq 
n\cdot \exp \left(- \frac{\delta^2m^2/2}{\Var[s(x)]+\delta w(d_{k}) m/3}\right) \leq \\
\leq n\cdot \exp\left(- \frac{\delta^2m}{3w(d_{k})}\right).
\end{split}
\end{equation}
and similarly,
\begin{equation}\label{eq:rdel-uneq}
\begin{split}
|\mathcal{R}(\delta m)| = n \cdot  \prob\left[s(x)\leq m-\delta m \right]
\leq n\cdot \exp \left(- \frac{\delta^2m}{3w(d_{k})}\right).
\end{split}
\end{equation}
Now we want to bound the number of entries in $T_{c}(\delta m, y)$ and $T_{r}(\delta m, x)$. For $x\in [\vec{d}]$ define $S_x = \{i\in [n]\mid x_i = d_i\}$. 

Let $x\in [\vec{d}]\setminus \calR(\delta m)$ and $y\in T_{r}(\delta m, x)$, by definition, $S_x\subseteq S_y$ and 
\begin{equation}\label{eq:sx-sy}
\sum\limits_{i\in S_y\setminus S_x} w(d_i) \leq 2\delta m.
\end{equation}
Thus, $t:= |S_y\setminus S_x|\leq 2 \delta m/w(d_1) =: t_{\max}$. Assumption (i) implies that 
\[ t_{\max}\leq 2\delta K\ln(n)/(d_k\ln d_k).\]
At the same time, $t_{\max}\geq 2\delta\cdot ( k w(d_1)/d_k)/w(d_1) = 2\delta k/d_k$.
Note that $y\in T_{r}(\delta m, x)$ is uniquely defined by $S_y\setminus S_x$. Thus, by Lemma~\ref{lem:tail-bound},    
 
\begin{equation}
\begin{gathered}
T_{r}(\delta m, x) \leq 
 \sum\limits_{i=0}^{t_{\max}} \binom{k}{i}\leq 
\exp\left(t_{\max} \ln\left(\dfrac{\eee k}{t_{\max}}\right)\right) \leq \\
\leq \exp(\ln n \cdot  2\delta K\ln(2d_k/\delta)/(d_k\ln d_k)).
\end{gathered}
\end{equation} 
Similarly, for $y\in [\vec{d}]\setminus \calC(\delta m)$ we have $x\in T_{c}(\delta m, y)$ only if $S_x\subseteq S_y$ and Eq.~\eqref{eq:sx-sy} is satisfied. Clearly, $|S_y|\leq m(1+\delta)/w(d_1)\leq 2m/w(d_1)$, and $t = |S_y\setminus S_x|$ satisfies $t\leq t_{\max}$. Hence, by Lemma~\ref{lem:tail-bound},
\begin{equation}
\begin{split}
T_{c}(\delta m, y) & \leq \sum\limits_{i=0}^{t_{\max}} \binom{|S_y|}{t_{\max}} d_k^{i}\leq  d_k^{t_{\max}}\sum\limits_{i=0}^{t_{\max}} \binom{2m/w(d_1)}{t_{\max}} \leq 
\\ & \leq\exp(t_{\max} \ln(3 \delta^{-1})+t_{\max}\ln{d_k})\leq  \exp(\ln n \cdot  2\delta K\ln(3\delta^{-1})/d_k).
\end{split}
\end{equation} 
Therefore, there exists $\delta = \Theta\left((\eps/K)/\log(K/\eps)\right)$ such that
\[ \max(M_{r}(\delta m), M_{c}(\delta m))\leq n^{\eps/(2d_k)}\,. \]

\noindent  
Moreover, for such choice of $\delta$, by assumption (ii) and Eq.~\eqref{eq:cdel-uneq}-\eqref{eq:rdel-uneq}, we obtain
\begin{equation}
|\mathcal{C}(\delta m)|+|\mathcal{R}(\delta m)|\leq n\exp(-\delta^2L\ln(n)/(3 d_k\ln d_k))\leq n^{1-\gamma}.
\end{equation}
Thus, the claim follows from Lemma~\ref{lem:rig-general} and 
Lemma~\ref{lem:rigid-reduct-Gd-gen}.
\end{proof}

Observe that Theorem~\ref{thm:equal-sizes} is a special case of
Theorem~\ref{thm:unequal} 
with $K = L = 1$. More generally, we can get the same bound on $\gamma$ for the case when all the $d_i$ are within a constant factor
of  
each other.

\begin{corollary}\label{cor:linear-gap}
Given $0<\eps, c \leq 1$, and $d\geq 2$, there exists 
$\gamma = \Omega\left(\dfrac{c\eps^2}{d\log(d)\log(1/c)\log^{2}(1/(c\eps))}\right)$
such that the following holds for any sequence of matrices
$M_1, M_2, \ldots, M_k$ where $M_i \in \fff^{d_i\times d_i}$
and $cd \le d_i\le d$.
Let $M = \bigotimes\limits_{i=1}^{k} M_i$ and
$n = \prod\limits_{i=1}^{k} d_i$.  
If $n \ge d^{1/\gamma}$, then
$ R^{rc}_{\fff}(M, n^{1-\gamma})\leq n^{\eps}$.
\end{corollary} 
\begin{proof} Observe that the assumptions of Theorem~\ref{thm:unequal} are satisfied for $w(x) = 1$ with $K = L \leq (1/c)\log(1/c)$. Let $\gamma'$ be the constant provided by Theorem~\ref{thm:unequal}. We take $\gamma = \gamma'/3$ and note that $n^{2\gamma}\geq 2d$. Hence, $R^{rc}_{\fff}(M, n^{1-\gamma})\leq n^{\eps}.$
\end{proof} 

Next, we eliminate the lower bound constraint on the $d_i$,
   at the cost of a slightly worse dependence of $\gamma$ on $d$.
   (Theorem~\ref{thm:row-column-main-noassum} below). In preparation for proving Theorem~\ref{thm:row-column-main-noassum}, we state its special case
   where all but at most one of the $d_i$ are restricted to the interval $[\sqrt{d}, d]$.
\begin{corollary}\label{cor:dtod2-bound}
Given $d\ge 2$ and $0<\eps <1$, there exists
$\gamma = \Omega\left(\dfrac{\eps^2}{d^{3/2}\log(d)\log^{2}(d/\eps)}\right)$
such that 
 the following holds. Consider any 
 $d_i\leq d$ and $M_i\in \fff^{d_i\times d_i}$ for $i\in [k]$.
 Let   $M = \bigotimes\limits_{i=1}^{k} M_i$ and
 $n = \prod\limits_{i=1}^{k} d_i$. Assume that
 $d_i\geq \sqrt{d}$ for all $i\geq 2$.
If $n\geq d^{1/\gamma}$.
Then $R^{rc}_{\fff}(M, n^{1-\gamma})\leq n^{\eps}$.
\end{corollary}
\begin{proof} Define $w(d_i) = 1$. Then, the assumptions of Theorem~\ref{thm:unequal} are satisfied for $K = L \leq 4\sqrt{d}$.  Indeed,  for $d_k\leq d$ and $d_2 \geq \sqrt{d}$ and $k>d$, we have the following trivial bounds
\begin{equation}
 \sum\limits_{i=1}^{k} \log(d_i)/\log(d_k)\geq (k-1)/2\quad  \text{and}\quad  \sum\limits_{i=1}^{k} 1/d_i\leq (k-1)/\sqrt{d}+1\leq 2(k-1)/\sqrt{d}. 
\end{equation}
Hence, by Theorem~\ref{thm:unequal}, there exists $\gamma' = \Omega\left(\dfrac{\eps^2}{d^{3/2}\log(d)\log^{2}(d/\eps)}\right)$
such that 
\[R^{rc}_{\fff}(M, 2d n^{1-\gamma'})\leq n^{\eps}.\]
Take $\gamma = \gamma'/3$. Finally, note that the assumption $n\geq d^{1/\gamma}$ implies that $n^{2\gamma}\geq 2d$ and $k\geq d$. Hence,
$R^{rc}_{\fff}(M, n^{1-\gamma})\leq n^{\eps}$. 
\end{proof}

Finally, we prove our main result by eliminating $d_i \ge \sqrt{d}$ constraint using
a bin packing argument.

\begin{theorem}\label{thm:row-column-main-noassum}
Given $d\ge 2$ and $0<\eps <1$, there exists
$\gamma = \Omega\left(\dfrac{\eps^2}{d^{3/2}\log(d)\log^{2}(d/\eps)}\right)$
such that the following holds for any sequence of matrices
$M_1, M_2, \ldots, M_k$ where $M_i \in \fff^{d_i\times d_i}$
and $2 \le d_i \le d$.
Let $M = \bigotimes\limits_{i=1}^{k} M_i$ and
$n = \prod\limits_{i=1}^{k} d_i$.
If $n\geq d^{1/\gamma}$, then $R^{rc}_{\fff}(M, n^{1-\gamma})\leq n^{\eps}$.
\end{theorem}
\begin{proof} 
Let us  split
the  $d_i$ into the smallest possible number $k'$ of bins, such that the product of numbers in each bin is at most $d$. Let $a_j$ be the product of
the  numbers in bin $j\in [k']$ and
let $A_j$ be the Kronecker product of the corresponding matrices. Then at most one $a_i$ is $\leq \sqrt{d}$.
Hence, the theorem follows from Corollary~\ref{cor:dtod2-bound}.
\end{proof}

\section{Application to rigidity upper bounds for Hadamard matrices}\label{sec:Hadamard}

As an application of our improved bounds, we show that the family of Kronecker products of sufficiently not rigid matrices and matrices of bounded size is not Valiant-rigid.

\begin{theorem}\label{thm:bounded-and-non-rigid} Let $0<\varepsilon <1/2$ and $b\geq 2$. Let $\mathcal{F}$ be a family of matrices over $\fff$, such that for every $d\times d$ matrix $A\in \mathcal{F}$ either $d\leq b$, or 
\begin{equation}
R_{\fff}^{rc}\left(A, d^{1-\gamma}\right) \leq d^{\varepsilon} \quad \text{for} \quad \gamma = \dfrac{12 (\log\log d)^2}{\varepsilon^3 \log d}.
\end{equation} 
Then, for every sequence of matrices $M_1, M_2, \ldots M_k \in\mathcal{F}$, the $n\times n$ matrix $M = \bigotimes\limits_{i\in [k]} M_i$ either satisfies
$
R_{\fff}^{rc}\left(M, {n}/{\log n}\right)\leq n^{6\varepsilon},
$  
 or $n$ is bounded above by a function of $b$ and $\varepsilon$.
\end{theorem}

 As an immediate application of Theorem~\ref{thm:bounded-and-non-rigid}, we significantly expand the family of Hadamard matrices known to be not Valiant rigid (Theorem~\ref{thm:main-Hadamard}). We rely on the rigidity upper bound for Paley-Hadamard matrices established in \cite{dvir-liu, field-matters}.

\begin{theorem}[\cite{dvir-liu, field-matters}]\label{thm:paley-hadamard-bound}
There exist constants $c_1>0$ and $c_2>0$ such that for all $\varepsilon>0$ and an arbitrary $d\times d$ Paley-Hadamard matrix $A$ we have
$
R_{\ccc}^{rc}\left(A, \dfrac{d}{\exp(\varepsilon^{c_1}(\log d)^{c_2})}\right)\leq d^{\varepsilon}
$.
\end{theorem}

We restate Theorem~\ref{thm:main-Hadamard} for convenience.

\begin{theorem} Let $\mathcal{F}_0$ be the family of Paley-Hadamard matrices and Hadamard matrices of bounded size. Let $\mathcal{F}$ be the family of all matrices that can be obtained as Kronecker products of some matrices from $\mathcal{F}_0$. Then $\mathcal{F}$ is not absolutely Valiant-rigid.
\end{theorem}
\begin{proof}
Fix $\varepsilon>0$. By Theorem~\ref{thm:paley-hadamard-bound}, the inequality
\[ \gamma = \dfrac{\varepsilon^{c_1}}{(\log d)^{1-c_2}}\geq \dfrac{12(\log\log d)^2}{\varepsilon^3 \log d}\quad  \Leftrightarrow \quad \dfrac{(\log d)^{c_2}}{(\log\log d)^2}\geq 12\varepsilon^{-3-c_1}\]
holds for all sufficiently large $d$. Thus, the claim of the theorem follows from Theorem~\ref{thm:bounded-and-non-rigid}.
\end{proof}

In order to prove Theorem~\ref{thm:bounded-and-non-rigid}, we use the following inequality that relates the row-column rigidity of the Kronecker product of a pair of matrices to the row-column rigidities of each of the matrices participating in the product.

\begin{lemma}[Dvir, Liu {\cite[Lemma 4.9]{dvir-liu}}]\label{lem:ten-prod-bound} Let $A$ be an $n\times n$ matrix and $B$ be an $m\times m$ matrix. Then
\[ R_{\fff}^{rc}(A\otimes B, r_a m+r_b n)\leq R_{\fff}^{rc}(A, r_a)\cdot R_{\fff}^{rc}(B, r_b)\]
\end{lemma}

In the pair of lemmas below we show that the Kronecker product of arbitrary many sufficiently large and sufficiently not rigid matrices is sufficiently not rigid itself. These lemmas are essentially the proof content of Lemma 4.10 in~\cite{dvir-liu}, which was stated and proved for more specific needs. 

\begin{lemma}\label{lem:ten-betw-squares} Let $2\leq  b\leq d_1\leq \ldots \leq d_k \leq b^2$ be integers. Let $0<\eps<1$ and  $\gamma >0$ such that $\gamma\log (b)\geq 2\log(3/\varepsilon)$ and $b\geq 3/\varepsilon$. Assume that for every $i\in [k]$, $M_i$ is a $d_i\times d_i$ matrix over $\fff$ that satisfies
\[ R_{\fff}^{rc}(M_i, d_i^{1 - \gamma}) \leq d_i^{\varepsilon}. \] 
Define $M = \bigotimes\limits_{i\in [k]} M_i$ and $n = \prod\limits_{i\in [k]} d_i$. Then, 
\[R_{\fff}^{rc}\left(M, n^{1 - \varepsilon \gamma/4}\right)\leq n^{4\varepsilon}.\] 
\end{lemma}
\begin{proof}
By the assumptions of the lemma, we can write $M_i = A_i+E_i$, where $A_i, E_i \in \fff^{d_i\times d_i}$, $\rank(A_i)\leq d_i^{1-\gamma}$, and $E_i$ has at most $d_i^{\varepsilon}$ non-zero entries in every row and column. Then 
\begin{equation}
 \begin{split}
 \bigotimes\limits_{i\in [k]} M_i = & \bigotimes\limits_{i\in [k]} (A_i+E_i) = \sum\limits_{S\subseteq [k]} \bigotimes\limits_{i\in S} A_i\bigotimes\limits_{j\in [k]\setminus S} E_j 
 \end{split}
\end{equation}
Now all the summands can be split into two groups: when $|S|\geq \varepsilon k$ and when $|S|< \varepsilon k$. We are going to bound the rank of the sum of the first group and the number of non-zero entries of the sum of the second group.
\begin{equation}
\begin{split}
\rank\left(\sum\limits_{S\subseteq [k], |S|\geq \varepsilon k} \bigotimes\limits_{i\in S} A_i\bigotimes\limits_{j\in [k]\setminus S} E_j \right)\leq \sum\limits_{S\subseteq [k], |S| = \varepsilon k}\rank\left( \bigotimes\limits_{i\in S} A_i\right)\prod\limits_{j\in [k]\setminus S} d_j \leq \\
\leq \binom{k}{\varepsilon k}\max\limits_{S\subseteq [k], |S| = \varepsilon k} \prod\limits_{i\in S} d_i^{1-\gamma}\prod\limits_{j\in [k]\setminus S} d_j\leq \left(\dfrac{3}{\varepsilon}\right)^{\varepsilon k} \cdot n\cdot b^{-\varepsilon k \gamma}
\end{split}
\end{equation}
Since $b^{\gamma}\geq (3/\varepsilon)^2$ and $n\leq b^{2k}$, we have
\begin{equation}
\rank\left(\sum\limits_{S\subseteq [k], |S|\geq \varepsilon k} \bigotimes\limits_{i\in S} A_i\bigotimes\limits_{j\in [k]\setminus S} E_j \right)\leq n\cdot b^{-\varepsilon k \gamma/2}\leq n^{1 - \varepsilon \gamma/4}
\end{equation}
Next, consider the remaining terms
\begin{equation}
E = \sum\limits_{S\subseteq [k], |S|< \varepsilon k} \bigotimes\limits_{i\in S} A_i\bigotimes\limits_{j\in [k]\setminus S} E_j. 
\end{equation}
Using Lemma~\ref{lem:tail-bound}, every column and every row of $E$ has at most \begin{equation}
\left(\sum\limits_{i = 0}^{\varepsilon k} \binom{k}{i}\right) (b^2)^{\varepsilon k} \prod_{i \in [k]} d_i^{\varepsilon} \leq \left(\dfrac{3}{\varepsilon}\right)^{\varepsilon k} n^{2\varepsilon}n^{\varepsilon}\leq b^{\varepsilon k}n^{3\varepsilon}\leq n^{4\varepsilon}
\end{equation}
 non-zero entries.
\end{proof}

\begin{lemma}\label{lem:general-ten-bound}
Let $2\leq b\leq d_1\leq \ldots \leq d_k $ be integers. Let $0<\eps<1$ and  $\gamma_1, \gamma_2, \ldots, \gamma_k >0$ be such that $b\geq 3/\varepsilon$ and for all $i\in [k]$ we have $\gamma_i \cdot \log d_i \geq 4\log(3/\varepsilon)$. Assume that for every $i\in [k]$, $M_i$ is a $d_i\times d_i$ matrix over $\fff$, that satisfies 
$ R_{\fff}^{rc}(M_i, d_i^{1 - \gamma_i}) \leq d_i^{\varepsilon}$.

Define $M = \bigotimes\limits_{i\in [k]} M_i$ and $n = \prod\limits_{i\in [k]} d_i$. Then, 
\[ R_{\fff}^{rc}\left(M, n^{1 - \psi}\right)\leq n^{5\varepsilon}\quad \text{for} \quad  \psi = \dfrac{\varepsilon^2\min_{i}\gamma_i}{4\log\log d_k}-\dfrac{\log\log\log d_k}{\log n}\] 
\end{lemma}
\begin{proof} 
Let $I_t = \{ i\mid d_i \in (b^{2^t}, b^{2^{t+1}}]\}$ for $t = 1, \ldots, L = \log\log d_k$. Let $A_t = \bigotimes\limits_{i\in I_t} M_i$ and $n_t = \prod\limits_{i\in I_t} d_i$. Let $\gamma = \min_i \gamma_i$. By Lemma~\ref{lem:ten-betw-squares},
\begin{equation}
R_{\fff}^{rc}(A_t, n_t^{1 - \varepsilon \gamma/4})\leq n_t^{4\varepsilon}
\end{equation}
Let $S = \{t\mid n_t \geq n^{\varepsilon / L}\}$ and $N_S = \prod\limits_{t\in S} n_t$. Then, by Lemma~\ref{lem:ten-prod-bound}, 
\begin{equation}
R_{\fff}^{rc}\left(\bigotimes\limits_{t\in S} A_t ,\  N_S\left(\sum\limits_{t\in S} n_t^{-\varepsilon \gamma/4}\right)\right)\leq N_S^{4\varepsilon}.
\end{equation}
Hence,
\begin{equation}
R_{\fff}^{rc}\left(\bigotimes\limits_{t\in S} A_t ,\  L \cdot N_S n^{-\varepsilon^2\gamma/(4L)}\right)\leq N_S^{4\varepsilon}
\end{equation}
Observe that, $n/N_S\leq \left(n^{\varepsilon/L}\right)^L\leq n^{\varepsilon}$. Thus, $R_{\fff}^{rc}\left(\bigotimes\limits_{t\in [L]\setminus S} A_t ,\ 0\right) \leq n^{\varepsilon} $.
Hence, by Lemma~\ref{lem:ten-prod-bound},
\begin{equation}
R_{\fff}^{rc}\left(M, n^{1 - \psi}\right)\leq n^{5\varepsilon},
\end{equation}
as $n^{-\psi} = \log\log d_k \cdot n^{-\varepsilon^2\gamma/(4\log\log d_k)} = Ln^{-\varepsilon^2\gamma/(4L)}$. 

\end{proof}

Finally, we are ready to prove Theorem~\ref{thm:bounded-and-non-rigid}.

\begin{proof}[Proof of Theorem~\ref{thm:bounded-and-non-rigid}]
 Define $b_\varepsilon = 3/\varepsilon$ and $b_* = \max(b, b_{\varepsilon})$. Let $\gamma_b =c_0\dfrac{1}{b_{*}^{3/2}\log^3(b_*)}\cdot\dfrac{\eps^2}{\log^2(1/\eps)}$, where $c_0>0$ is the constant given by Theorem~\ref{thm:row-column-main-noassum}. We may assume $c_0<1$. Define $N_b = b_{*}^{24/(\varepsilon \gamma_b)}$. 

Denote the size of $M_i$ by $d_i$ for $i\in [k]$ and let $S = \{ i\in [k] \mid d_i \leq b_{*}\}$. Let 
\begin{equation}
F = \bigotimes\limits_{i\in S} M_i \quad \text{and}\quad  H = \bigotimes\limits_{j\in [k]\setminus S} M_j. 
\end{equation} 
Denote by $N_F$ and $N_H$ be the orders of $F$ and $H$, respectively. Then $N_F\cdot N_H = n$. Let $d_{\max} = \max_i d_i$. By Lemma~\ref{lem:general-ten-bound},
\begin{equation}
R_{\fff}^{rc}\left(M, N_H^{1 - \psi_H}\right)\leq N_H^{5\varepsilon},
\end{equation}
where
\begin{equation}\psi_H \geq \dfrac{\varepsilon^2}{4\log\log d_{\max}}\cdot \dfrac{12(\log\log d_{\max})^2}{\varepsilon^3 \log d_{\max}} - \dfrac{\log\log\log d_{\max}}{\log N_H}\geq \dfrac{2\log\log n}{\varepsilon \log n}.
\end{equation}
At the same time, by Theorem~\ref{thm:row-column-main-noassum}, if $N_F\geq N_b$, then
\begin{equation}
R_{\fff}^{rc}\left(F, N_F^{1-\gamma_b}\right)\leq N_F^{\varepsilon}.
\end{equation}
Note, since $b_{*}\geq 3/\varepsilon$, we have 
\begin{equation}
     \log\log N_b\leq \log( 1/\gamma_b)+2\log(b_{*})\leq 12\log(b_{*}) \leq \gamma_b\varepsilon \log N_b/2.  
\end{equation}
So, in the case $n\geq N_F\geq N_b$ we have $\gamma_b\geq \dfrac{2\log\log n}{\varepsilon \log n}$.

\noindent If $\min(N_H, N_F)\geq n^{\varepsilon}$ and $N_F\geq N_b$, then, by Lemma~\ref{lem:ten-prod-bound},
\begin{equation}
\begin{split}
& R^{rc}_{\fff}\left( M, n/\log n\right) \leq  
R^{rc}_{\fff}\left( M, n(n^{-\varepsilon\psi_H}+n^{-\varepsilon\gamma_b})\right) \leq \\
& \quad \leq  R^{rc}_{\fff}\left( H\cdot F, n(N_H^{-\psi_H}+N_F^{-\gamma_b})\right) \leq  R_{\fff}^{rc}\left(H, N_H^{1 - \psi_H}\right) \cdot R_{\fff}^{rc}\left(F, N_F^{1 - \gamma_b}\right) \leq n^{5\varepsilon}.
\end{split}
\end{equation}
If $N_F\leq n^{\varepsilon}$, then $N_H\geq n^{\varepsilon}$, and by Lemma~\ref{lem:ten-prod-bound}, 
\begin{equation}
R^{rc}_{\fff}\left( M, n/\log n\right) \leq R_{\fff}^{rc}\left(M, N_FN_H^{1 - \psi_H}\right)\leq N_H^{5\varepsilon}\cdot R_{\fff}^{rc}\left(F,0\right)\leq N_H^{5\varepsilon}\cdot N_F \leq n^{6\varepsilon}. 
\end{equation}
Similarly, if $N_H\leq n^{\varepsilon}$ and $N_F\geq N_b$, then $R^{rc}_{\fff}\left( M, n/\log n\right) \leq n^{6\varepsilon}$.

Finally, if $N_b > N_F\geq n^{\varepsilon}$, then the size of $M$ is bounded by a function of $b$ and $\varepsilon$.
\end{proof}


\appendix

\section{Proof of Observation~\ref{obs:factor}}\label{sec:proof-obs-factor}

In this appendix we prove Observation~\ref{obs:factor}.

\begin{observation}\label{obs:factor-app} For any matrix $A \in \fff^{d\times d}$ there exists $B\in \fff^{(d-1)\times (d-1)}$, vectors $x, y\in \mathbb{F}^d$, $\lambda\in \{0, 1\}\subseteq \fff$ and permutation matrices $P_1, P_2\in \fff^{d\times d}$ such that
\begin{equation}\label{eq:factor-app} A = P_1 \cdot  G_d(y)^T  \cdot \left(\begin{matrix}
		B & 0\\
		0 & \lambda
		\end{matrix}\right)   \cdot G_d(x)  \cdot P_2.
\end{equation}
\end{observation}
\begin{proof}
We consider two cases. First, assume that $A$ has rank $d$. In this case, the basis vector $e_d$ can be written as a linear combination of the columns $A_i$ of $A$. Moreover, by changing the order of columns (by some permutation matrix, say $Q_1$) we may assume that the coefficient in front of the last column is non-zero. In other words, there exist coefficients $\mu_1, \ldots, \mu_d$, with $\mu_d\neq 0$ such that
\[ e_d = \mu_1 (AQ_1)_1+\mu_2(AQ_1)_2 + \ldots +\mu_d (AQ_1)_d. \]
Define $x_i = -\mu_i/\mu_d$ for $i\in [d-1]$ and $x_d = \dfrac{1}{\mu_d}$. Denote by $A'$ the matrix consisting of the first $d-1$ columns of $AQ_1$. Then 
\begin{equation}
 AQ_1 = \left( A'\  \begin{matrix}
0_{d-1}\\
1
\end{matrix}\right)G_d(x).
\end{equation}
Since $A$ has full rank, the first $d-1$ rows of $A'$ span $\fff^{d-1}$. So the last row of $A'$ can be written as a linear combination of the first $d-1$ rows. Hence, for some vector $y\in \fff^d$ with $y_d = 1$,  
\begin{equation}
(A')^T = \left( B\ 0_{d-1}\right)G_d(y) \quad \Rightarrow \quad AQ_1 =   G_d(y)^T  \cdot \left(\begin{matrix}
		B & 0\\
		0 & 1
		\end{matrix}\right)   \cdot G_d(x). 
\end{equation}  
If $A$ has rank less than $d$, there exists a column and a row of $A$ that can be expressed as a linear combination of all other columns and rows of $A$, respectively. By changing the order of rows and columns we may assume that these are $d$-th column and $d$-th row. Then, similarly as above, we see that $A$ can be written in the form~\eqref{eq:factor-app} with $\lambda = 0$.   
\end{proof}

\bibliography{mfcs-kivva}

\end{document}